\pdfoutput=1 
\documentclass[a4paper,USenglish,cleveref]{lipics-v2021}
\usepackage{cite}
\nolinenumbers 
\hideLIPIcs  
\usepackage{mathtools}
\usepackage{algorithm}
\usepackage{graphicx}
\usepackage[noend]{algpseudocode}

\graphicspath{{./figures/}}

\title{Bandwidth vs BFS Width in Matrix Reordering, Graph Reconstruction, and Graph Drawing}

\titlerunning{Bandwidth vs BFS Width}

\author{David Eppstein}{University of California, Irvine, USA}{eppstein@uci.edu}{}{Supported by NSF grant 2212129, as well for Michael Goodrich and Songyu (Alfred) Liu.}

\author{Michael T. Goodrich}{University of California, Irvine, USA}{goodrich@uci.edu}{https://orcid.org/0000-0002-8943-191X}{}

\author{Songyu (Alfred) Liu}{University of California, Irvine, USA}{songyul4@uci.edu}{https://orcid.org/0009-0003-1255-7156}{}

\authorrunning{D.~Eppstein, M.T.~Goodrich, and S.~Liu}

\Copyright{David Eppstein, Michael T. Goodrich, and Songyu (Alfred) Liu} 

  

\category{} 

\relatedversion{} 

\EventEditors{Anne Benoit, Haim Kaplan, Sebastian Wild, and Grzegorz Herman}
\EventNoEds{4}
\EventLongTitle{33rd Annual European Symposium on Algorithms (ESA 2025)}
\EventShortTitle{ESA 2025}
\EventAcronym{ESA}
\EventYear{2025}
\EventDate{September 15--17, 2025}
\EventLocation{Warsaw, Poland}
\EventLogo{}
\SeriesVolume{351}
\ArticleNo{67}

\renewcommand{\emph}[1]{\textbf{\textit{#1}}}
\newcommand{\bfsw}{{\mathbf{bfsw}}}
\newcommand{\bw}{{\mathbf{bw}}}

\newcommand{\tw}{{\mathbf{tw}}}
\newcommand{\polylog}{{\mathrm{polylog}}}

\begin{document}

\maketitle

\begin{abstract}
We provide the first approximation quality guarantees for the Cuthull-McKee heuristic for reordering symmetric matrices
to have low bandwidth, and we provide an algorithm for reconstructing bounded-bandwidth
graphs from distance oracles with near-linear query complexity. To prove these results we introduce a new width parameter, BFS width,
and we prove polylogarithmic upper and lower bounds on the BFS width of graphs of bounded bandwidth.
Unlike other width parameters, such as bandwidth, pathwidth, and treewidth, BFS width can easily be computed in polynomial time.
Bounded BFS width implies bounded bandwidth, pathwidth, and treewidth, which in turn imply fixed-parameter tractable algorithms
for many problems that are NP-hard for general graphs. 
In addition to their applications to matrix ordering,
we also provide applications of BFS width to graph reconstruction,
to reconstruct graphs from distance queries,
and graph drawing, to construct arc diagrams of small height.
\end{abstract}

\section{Introduction}
When a graph has a linear layout of low bandwidth, it is natural
to guess that breadth-first search produces a low-width layout.
This is the underlying principle of the widely used Cuthill–McKee
algorithm~\cite{cuthill_reducing_1969} for reordering symmetric
matrices into band matrices with low bandwidth. In this paper, we
quantify this principle, for the first time, by providing  the first
worst-case analysis of the Cuthill–McKee algorithm, proving
polylogarithmic upper and lower bounds for its performance on
bounded-bandwidth graphs. Using the same techniques, we also study
the problem of reconstructing an unknown graph, given access to a
distance oracle. We use a method based on breadth-first search to
reconstruct graphs of bounded bandwidth in near-linear query
complexity. Additionally, we apply our results in graph drawing,
by constructing arc diagrams of low height for any graph of bounded
bandwidth.

The key to our results is a natural but previously-unstudied graph width parameter, the \emph{breadth-first-search width} or \emph{BFS width} of a graph, which we define algorithmically in terms of the size of the layers in a breadth-first search tree. To develop our results, we prove that bandwidth and BFS width are polylogarithmically tied: the graphs of bandwidth $\le b$, for any constant $b$, have BFS width that is both upper bounded (for all graphs of bandwidth $\le b$) and lower bounded (for infinitely many graphs of bandwidth $b$) by functions of the form $\log^{f(b)}n$.
To provide context for our contributions, 
we review the background of our three application areas below and
then describe our new results in more detail for each.

\paragraph*{Sparse Numerical Linear Algebra}
A symmetric matrix, $A$, is called a \emph{band matrix} if its
nonzeros are confined to a \emph{band} of consecutive diagonals
centered on the main diagonal; the \emph{bandwidth} is the number
of these nonzero diagonals on either side of the main diagonal.
That is, it is the maximum value of $|i-j|$ for the indexes of a
nonzero matrix entry $A_{i,j}$. Any symmetric matrix corresponds
to an undirected graph, with $n$ numbered vertices, and with an
edge $\{i,j\}$ for each off-diagonal nonzero entry $A_{i,j}$; the
numbering of vertices describes a \emph{linear layout} of this graph
and the bandwidth of this layout is the largest difference between
the positions of any two adjacent vertices in this linear layout.
The bandwidth of the graph itself is the minimum bandwidth of any
of its linear layouts; the goal of the Cuthill–McKee algorithm is
to permute the vertices of the graph into a layout with bandwidth
close to this minimum, and correspondingly to permute the rows and
columns of the given matrix to produce an equivalent symmetric band
matrix of low bandwidth. 
See Figure~\ref{fig:cuthill}.
Such a reordering, for instance, can be
used in sparse numerical linear algebra to reduce fill-in (zero
matrix elements that become nonzero) when applying Gaussian elimination
to the matrix.

\begin{figure}[hbt]
\centering
\begin{tabular}{c@{\hspace*{.5in}}c}
\fbox{\includegraphics[width=2.14in]{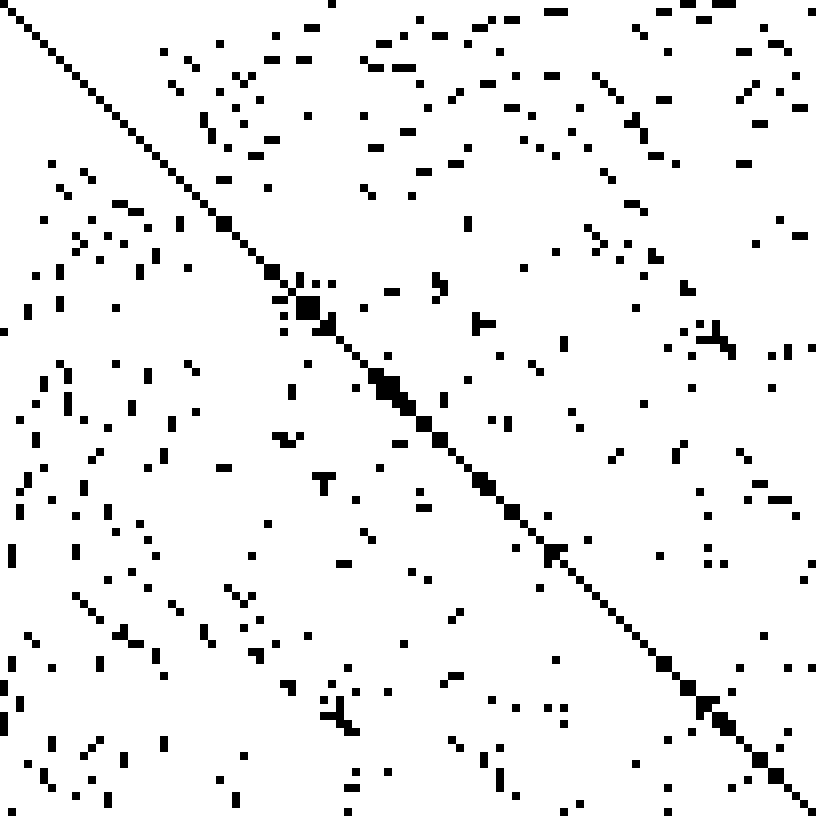}} &
\includegraphics[width=2.25in, trim=.8in 1.8in 0in 0.1in, clip]{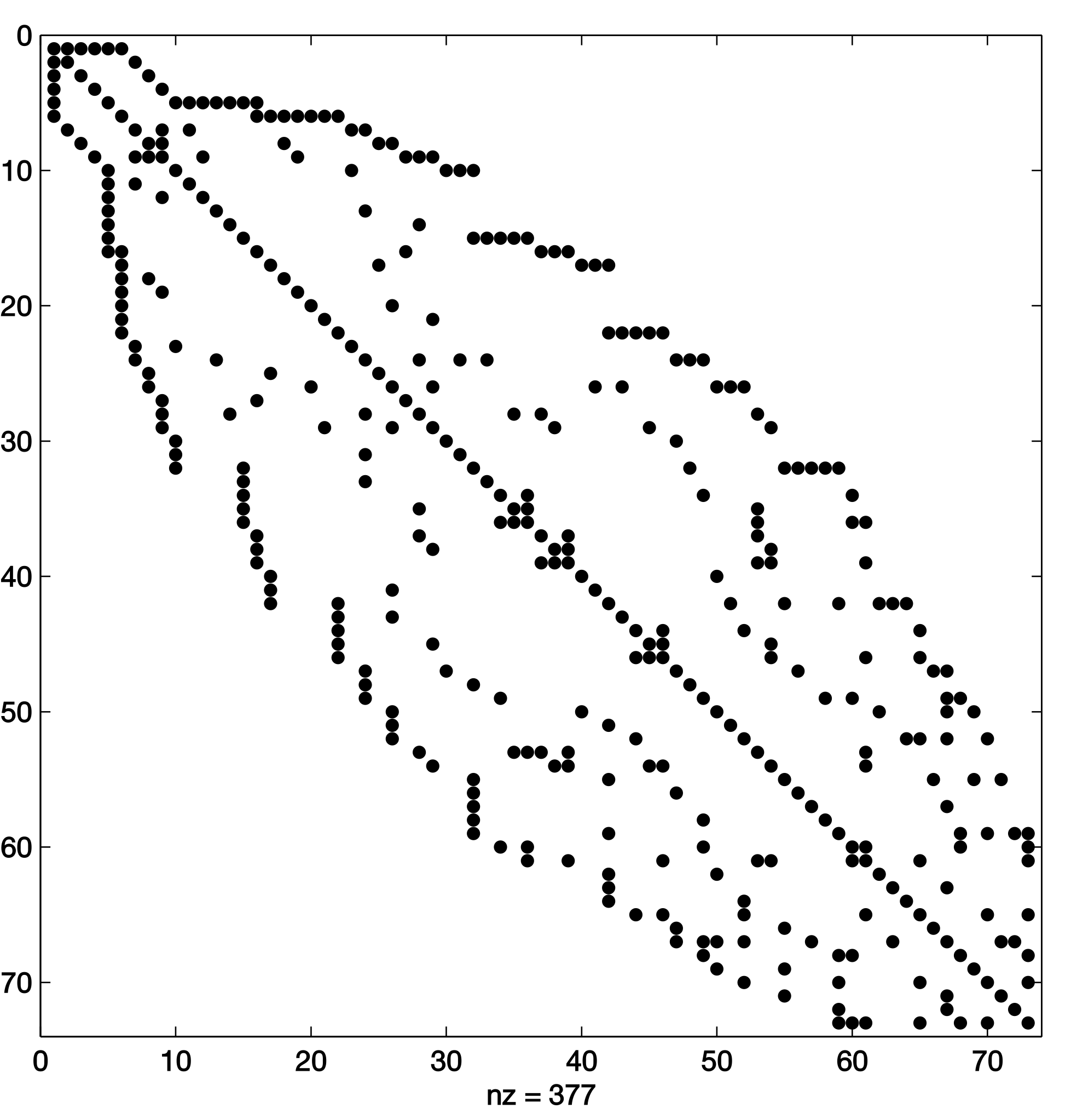} \\
(a) & (b)
\end{tabular}
\caption{Illustrating the Cuthill-McKee algorithm. (a) A sparse matrix;
(b) a sparse matrix reordered according to the 
Cuthill-McKee algorithm. 
Nonzero entries are shown as black spots.
Left: public-domain image by Oleg Alexandrov;
right: image by Wikipedia user Kxx licensed under the CC BY-SA 3.0 license.
\label{fig:cuthill}
}
\end{figure}

Graph bandwidth is NP-complete~\cite{GarJoh-79} and hard to approximate to within any constant factor, even for trees of a very special form~\cite{dubey_hardness_2011}.  For any fixed $b$, testing whether bandwidth $\le b$ is polynomial, but with $b$ in the exponent of the polynomial~\cite{GurSud-JAlgs-84}. Therefore, practitioners have resorted to heuristics for bandwidth, rather than exact optimization algorithms. Notable among these are the Cuthill–McKee algorithm~\cite{cuthill_reducing_1969}, and the closely related reverse Cuthill–McKee algorithm~\cite{GeoLiu-81}; both are heavily cited and widely used for this task. More sophisticated methods have also been used and tested experimentally~\cite{SafRonBra-JEA-09}.

The best known polynomial time algorithm provides a $ O(\log^3 n \sqrt{\log \log n})$-approximation factor for bandwidth on general graphs with high probability using a semi-definite relaxation \cite{dunagan_euclidean_2001}. 
A randomized polynomial time  $O(\log^{2.5} n)$-approximation algorithm is known for the special cases of trees and chordal graphs using caterpillar decomposition~\cite{gupta_improved_2001}. 
 In contrast, the Cuthill–McKee and reverse Cuthill–McKee algorithms are much simpler since they only rely on BFS and are deterministic.
Despite their simplicity, we are unaware of any past work providing worst-case performance guarantees for  the Cuthill–McKee and reverse Cuthill–McKee algorithms. 
The only theoretical upper bound we are aware of regarding these algorithms concerns random matrices whose optimal bandwidth is sufficiently large, for which this algorithm will produce a constant factor approximation. The same work also provides a logarithmic lower bound on the approximation factor for matrices with optimal bandwidth 2, weaker than our result~\cite{turner_probable_1986}. 
Instead, we are interested in deterministic guarantees for these algorithms for matrices of low optimal bandwidth. We prove the following results:
\begin{itemize}
\item For $n\times n$ matrices of bounded optimal bandwidth $\le b$, we show
that the Cuthill–McKee algorithm and the reverse Cuthill–McKee algorithm will find a reordering of bandwidth upper bounded by a function of the form $\log^{f(b)} n$. Equivalently, for $n$-vertex graphs of bandwidth $\le b$, these algorithms will find a linear layout of bandwidth $\le \log^{f(b)} n$.
\item We show that there exist matrices for which this polylogarithmic dependence is necessary. For every polylogarithmic bound $\log^{k - 1} n$ on the bandwidth, there is a constant bandwidth bound $b$ such that, for infinitely many $n$, there exist matrices of optimal bandwidth $\le b$ that cause both the Cuthill–McKee algorithm and the reverse Cuthill–McKee algorithm to produce reordered matrices of bandwidth $\Omega(\log^{k - 1} n)$.
\end{itemize}

We remark that, for any graph, 
\[
\mbox{
$O(1)$ BFS width $\implies$ 
$O(1)$ bandwidth $\implies$ $O(1)$ pathwidth $\implies$ $O(1)$ 
treewidth~\cite{KapSha-SICOMP-96}, 
}
\]
and graphs of 
bandwidth $b$ have tree-depth $O(b\log n/b)$~\cite{Gru-JComb-12}. 
Thus, many algorithmic results based on width parameters such as
pathwidth, treewidth, and tree-depth are also available for BFS-width.
However, both pathwidth and tree-depth have approximation algorithms
with fixed polylogarithmic approximation ratios~\cite{FeiHajLee-STOC-05,BodGilHaf-JAlg-95}, better than the
polylogarithmic ratio with a variable exponent depending on the
optimal bandwidth that we prove for the Cuthill–McKee algorithm.

\paragraph*{Graph Reconstruction}

We also provide new results for \emph{graph reconstruction}, a well-motivated problem with numerous applications, such as learning road networks \cite{afshar_efficient_2022} and communication network mapping \cite{afshar_mapping_2022}. Suppose we are given the vertex set but not the edge set of a graph $G$.
The \emph{graph reconstruction} problem 
is to determine all the edges of $G$, by asking queries about $G$ from an oracle, with the goal of minimizing the \emph{query complexity}, the number of queries to the oracle \cite{kannan_graph_2018, mathieu_simple_2021}. Various oracles can be used. In particular, we consider the \emph{distance oracle} from several previous works \cite{kannan_graph_2018, mathieu_simple_2021}. This oracle takes a pair of vertices and returns the number of edges on a shortest path between the two given vertices. Following previous work~\cite{kannan_graph_2018, mathieu_simple_2021}, the graph is assumed to be connected, undirected, and unweighted. It is also generally assumed to have bounded degree, because distinguishing an $n$-vertex star from a graph with one additional edge between the leaves of the star would have a trivial quadratic lower bound. For reconstruction from distance oracles, the following results are known:
\begin{itemize}
\item For graphs of maximum degree $\Delta$, there is a randomized algorithm with query complexity $O (\Delta^3 \cdot n^{3/2} \cdot \log^2 n \cdot \log \log n)$, which is $\tilde{O} (n^{3/2})$ when $\Delta = O\bigl(\polylog(n)\bigr)$~\cite{kannan_graph_2018}.\footnote{We use $\tilde{O}(*)$ to denote
  asymptotic bounds that ignore polylogarithmic factors.}
\item Chordal graphs can be reconstructed in randomized query complexity $\tilde{O}(n)$ when $\Delta=O(\log\log n)$~\cite{kannan_graph_2018}. The $k$-chordal graphs can be reconstructed in randomized query complexity  $O_{\Delta,k}(n \log n)$ when  $\Delta = O(1)$ \cite{bastide_tight_2025}. 
\item Outerplanar graphs can be reconstructed in randomized query complexity $\tilde{O}(n)$ when $\Delta$ is polylogarithmic~\cite{kannan_graph_2018}.
\item For reconstructing uniformly random $\Delta$-regular graphs with maximum degree $\Delta = O(1)$, there is a randomized algorithm using $\tilde{O}(n)$ queries in expectation~\cite{mathieu_simple_2021}.
\end{itemize}
Our new results in this area are that the graphs of bounded bandwidth can be reconstructed deterministically in query complexity $\tilde{O}(n)$, by a very simple algorithm. More generally, the graphs of BFS width $B$ can be reconstructed deterministically in query complexity $O(nB)$. No additional assumption on $\Delta$ is needed here because bounded bandwidth and bounded BFS width both automatically imply bounded degree.

\paragraph*{Graph Drawing}

Graph drawing is the study of methods for visualizing 
graphs; see, e.g.,~\cite{battista_graph_1998,komarek}. 
One type of graph visualization tool that has received considerable attention
is the \emph{arc diagram} of a graph, 
$G=(V,E)$, 
where the vertices of $G$ 
are laid out as points on a straight line and edges are drawn 
as semicircular arcs or straight-line segments 
(for consecutive points) joining pairs of such points.
See Figure~\ref{fig:arcs}.

\begin{figure}[hbt]
\centering
\includegraphics[width=2.5in]{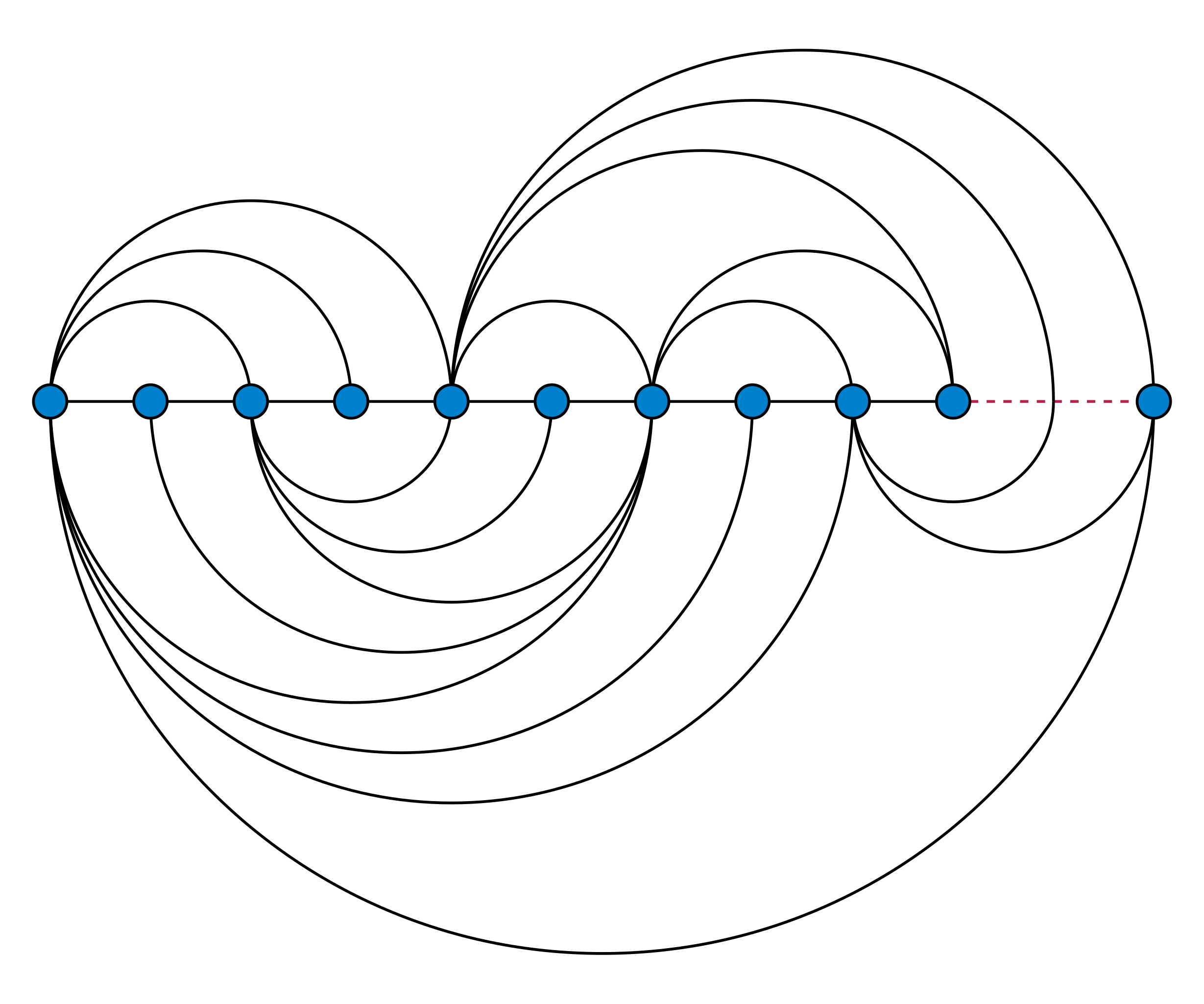} 

\caption{An example arc diagram.
Public-domain image by David Eppstein.
\label{fig:arcs}
}
\end{figure}

Arc diagrams have received considerable attention, both as a visualization
tool and as a topic of study in discrete mathematics; see,
e.g.,~\cite{masuda_crossing_1990,wattenberg_arc_2002,cardinal_et_al,burch2021dynamic,debiasi,komarek}. As a direct consequence of \cref{theorem:upper}, when an $n$-vertex graph has an arc diagram of bounded height, we can find a drawing with polylogarithmic height.
Here, the height of an arc diagram, with vertices placed at integer positions along the $x$-axis, is just half the bandwidth of the linear layout describing the vertex placement.

\section{Preliminaries}

We define BFS width as follows.
From a given undirected graph, $G=(V,E)$, choose arbitrarily a starting vertex $v$ in $G$
and perform a breadth-first search of $G$ from $v$. 
The edges traversed during this search define a breadth-first search tree $T$ such that the depth of each
vertex in $T$ equals its unweighted distance from $v$, the number of edges in a shortest path from $v$.
We define the \emph{layer} $L_i(v)$ to consist of all vertices at distance $i$ from $v$; if $u \in L_i(v)$, we say that $u$ has \emph{layer number} $i$.
The maximum size of any $L_i(v)$ set is the BFS width of $G$ from $v$
and the BFS width of $G$ is the maximum BFS width taken over all choices of the starting vertex $v$ in $G$.
That is, we have the following formal definition of BFS width building upon the definition of width in layered graph drawing \cite{battista_graph_1998}.

\begin{definition}[BFS Width]
    Let $G = (V, E)$ be a graph. A \emph{layering} of $G$ is a partition of $V$ into subsets $L_0, L_1, \dots, $ and the \emph{width} of $G$ for this layering is $\max_{i \in \mathbb{Z}_{\geq 0}} |L_i|$. Any vertex, $v\in V$, defines a layering by letting $L_i(v) = \{u \in V : d(v, u) = i\}$.
The \emph{BFS width of $G$ from $v$}, denoted $\bfsw(G, v)$, is defined as $\bfsw(G, v) = \max_{i \in \mathbb{Z}_{\geq 0}} |L_i(v)|$.
    The \emph{BFS width of $G$}, denoted $\bfsw(G)$, is defined as 
$\bfsw(G) = \max_{v \in V} \bfsw(G, v) = \max_{v \in V}  \max_{i \in \mathbb{Z}_{\geq 0}} |L_i(v)|$.
\end{definition} 

We also define the \emph{minimum BFS width},
$\bfsw_{\min}(G)=\min_{v\in V} \bfsw(G,v)$. These two BFS width parameters can differ significantly, and we will show that there exist graphs such that $\bfsw(G) / \bfsw_{\min}(G) = \Omega(\log n)$. 
Clearly, $\bfsw(G)$ and $\bfsw_{\min}(G)$
can both easily be computed in $O\bigl(n(n+m)\bigr)$ time
for any graph, $G$, with $n$ vertices and $m$ edges 
by performing $n$ breadth-first searches; 
see, e.g.,~\cite{Cormen2022Algorithms,goodrich}.
In contrast, it is NP-hard to compute other well-known width
parameters for a graph, including bandwidth~\cite{bandwidth2},
pathwidth, and treewidth~\cite{bodlaender_efficient_1996}.
Nevertheless, we prove that bounded BFS width 
implies bounded bandwidth for the graph, 
which, in turn, implies bounded pathwidth and treewidth for the graph.

\subsection*{Some Preliminary Facts About Bandwidth and BFS Width}
\paragraph*{Bandwidth}
Two lower bounds on bandwidth are as follows:
\begin{claim}[Degree lower bound  \cite{bottcher_bandwidth_2010}]
    $\bw(G) \geq \lceil\Delta/2\rceil$, where $\Delta$ is the graph's maximum 
degree.
\end{claim}
\begin{claim}[Local density lower bound \cite{gupta_improved_2001}]
    $\bw(G) \ge D(G) \coloneq \max_{v \in V, d} \left\lceil \frac{|N(v, d)|}{2d} \right\rceil$, where $D(G)$ is the local density of $G$, and $N(v, d)$ is vertices at distance at most $d$ from $v$ (excluding $v$).
\end{claim}
When $v$ is a maximum degree vertex and $d = 1$, this implies the degree lower bound.

\paragraph*{BFS Width}
Recall that a \emph{caterpillar tree} is a tree in which every 
vertex is within distance $1$ of a central path; i.e.,
the central path contains every vertex
of degree $2$ or more; see, e.g.,~\cite{HARARY1973359}.
Caterpillar trees have applications in chemistry and physics~\cite{el1987applications},
and they are the graphs of pathwidth one~\cite{proskurowski1999classes}.
A simple, but useful, fact regarding BFS width is the following:

\begin{theorem}
If $G$ is a caterpillar tree with maximum degree, $\Delta \geq 2$,
then $\bfsw(G)\le 2(\Delta-1)$.
\end{theorem}
\begin{proof}
To upper bound $\bfsw(G)$, we want to find a root vertex $v$ that can maximize the BFS width of $G$ from $v$. Choose a vertex, $v$,  that is not one of the end points on the central path as the root for a BFS tree, $T$. 
Then layer 1 will have size at most $\Delta$ and each  
vertex of degree $\Delta$ in $G$ that is not $v$ will correspond to an interior vertex in $T$ with at most $\Delta-1$ children. All degree $\Delta$ vertices are on the central path.
The maximum layer size occurs in $T$ when two such vertices are in the same layer. The layer below it has  $2(\Delta-1)$ children.
\end{proof}

We also have the following.

\begin{theorem}
\label{thm:forward}
For any graph, $G$, its bandwidth $\bw(G)\le 2\bigl(\bfsw_{\min}(G)\bigr)-1$.
\end{theorem}
\begin{proof}
Let $T$ be the BFS tree for $G$
rooted at a vertex, $v$, that achieves the minimum BFS width,
$\bfsw_{\min}(G) =\bfsw(G,v)$.
By definition,
every layer in $T$ has at most $\bfsw_{\min}(G)$ vertices.
It is well-known
that in a graph, $G$, with a BFS tree, $T$, every non-tree edge, $(u,v)$,
connects two vertices whose layer numbers differ by at most $1$; 
see, e.g.,~\cite{goodrich,Cormen2022Algorithms}.
Thus, we can number the vertices of $T$ as $1,2,\ldots,n$ so as to have 
nondecreasing layer numbers. The BFS ordering is one way to achieve this.
This numbering defines a linear layout of the vertices. The maximum distance on the $x$-axis between two endpoints of an edge, therefore, will be at most
that determined by the first vertex in a layer $i$ and the last
vertex in a layer $i+1$, which can be at most $2\bigl(\bfsw_{\min}(G)\bigr)-1$.
\end{proof}

This immediately implies the following.

\begin{corollary}
\label{cor:forward}
For any graph, $G$, $\bw(G)\le 2\bigl(\bfsw(G)\bigr)-1$.
\end{corollary}

\begin{figure}[ht]
    \centering
    \includegraphics[width=1\linewidth]{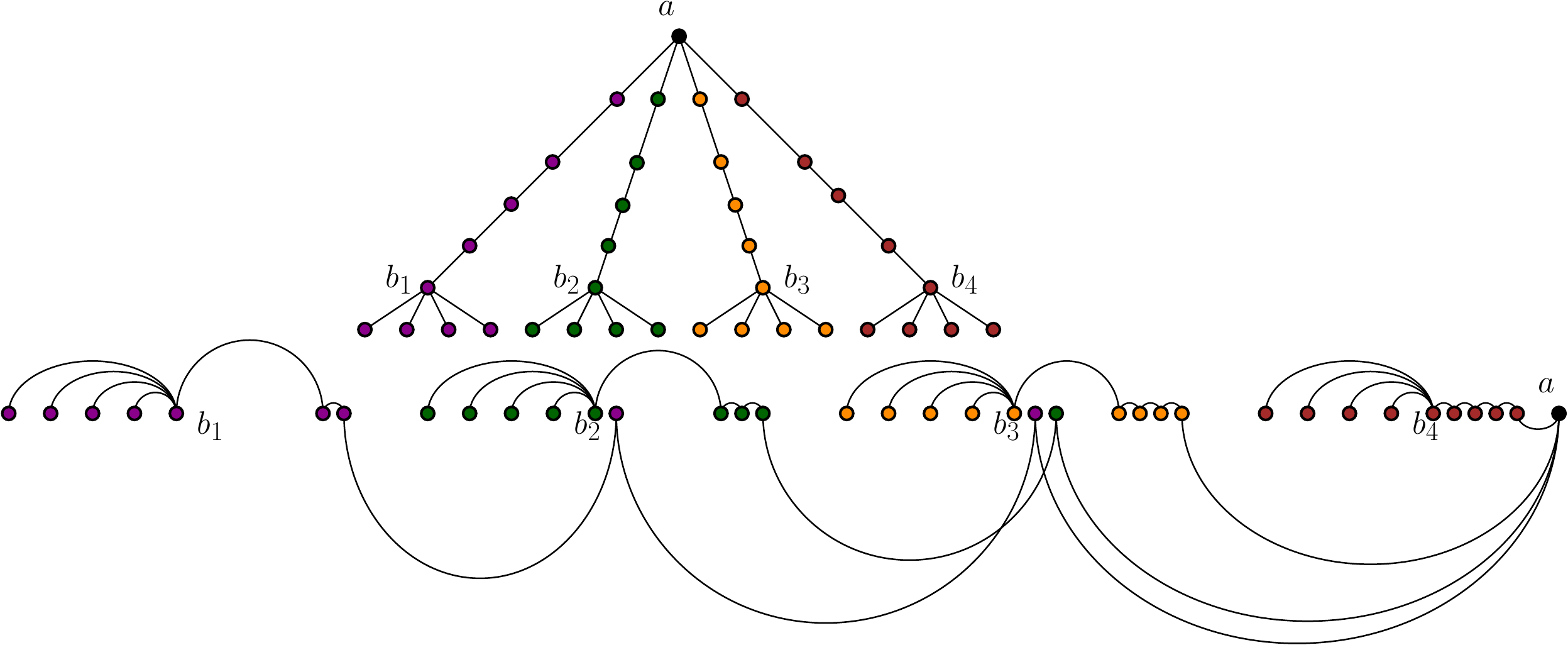}
    \caption{An example graph with $\bfsw(G)=\Theta(n)$ and $\bfsw_{\min}(G)=\Theta(n)$ and a linear layout that achieves its optimal bandwidth $\bw(G)=\Theta(\sqrt{n})$.}
    \label{fig:bw-sqrt-n}
\end{figure}
Given a linear layout of the vertices of a graph with respect to
bandwidth, we say that the \emph{length} of an edge, $\{u,v\}$, is
one plus the number of vertices between $u$ and $v$.  
\begin{theorem}
There is a graph, $G$, with $n$ vertices
such that its bandwidth $\bw(G)=\Theta(\sqrt{n})$ but $\bfsw(G)=\Theta(n)$ and $\bfsw_{\min}(G)=\Theta(n)$.
\end{theorem}
\begin{proof}
    The example is illustrated in \cref{fig:bw-sqrt-n}. 
    Take a star with $\sqrt{n}$ leaves and subdivide each edge $\sqrt{n}$ times. Then add a star with $\sqrt{n}$ leaves at the end of each path. This example first appeared in \cite{gupta_improved_2001} but they only gave a lower bound on bandwidth and did not compute the bandwidth. By the degree lower bound, $\bw(G) = \Omega(\sqrt{n})$. Regardless of which vertex is chosen as the root for BFS, BFS will produce a layout with bandwidth $\Theta(n)$. Now it remains to show a linear layout with bandwidth  $ O(\sqrt{n})$.

The path from $a$ to vertex $b_i$ and the star at $b_i$ are considered as one block and we call this block $B_i$. Suppose we have already placed $B_{\sqrt{n}-1}, B_{\sqrt{n}-2}, \dots, B_{i+1}$.  Now we describe how to place $B_i$. All   $\sqrt{n}$ leaves of  $b_i$ are placed on its left. Next we place the path from $a$ to $b_i$. Starting from $a$, we place one vertex to the right of $b_{\sqrt{n}-1}, b_{\sqrt{n}-2}, \dots, b_{i+1}$. This step uses $\sqrt{n} - i - 1$ vertices. The remaining path segment with $i + 1$ vertices are placed to the right of $b_i$.
    
    Next we analyze the bandwidth. The edge between $b_i$ and  its leftmost leaf  has length $\sqrt{n}$. Any vertex $b_i$ introduces a path vertex to the right of $b_{\sqrt{n}-1}$. In total there are $\sqrt{n} - 2$ vertices between  $b_{\sqrt{n}-1}$ and its parent in the tree. Block $B_{\sqrt{n}}$ has $2\sqrt{n} + 1$  vertices. Thus, the length of the first edge on path $ab_1$ is at most $4\sqrt{n}$. The  linear layout  has bandwidth  $O(\sqrt{n})$.
\end{proof}
This shows that there are graphs for which the linear layout produced by BFS is off by a factor of $\Theta(\sqrt{n})$.

As noted above, there are many algorithmic applications
with respect to graphs with bounded bandwidth, which,
in turn, implies bounded pathwidth and treewidth.
Theorem~\ref{thm:forward} implies that a graph with bounded minimum BFS
width has bounded bandwidth. 
Thus, 
any polynomial-time algorithm
for graphs with bounded bandwidth, pathwidth, or treewidth 
(e.g., see~\cite{andreica_dynamic_2012, bodlaender_combinatorial_2008}) also applies to any graph with bounded BFS width
for one of its vertices,
with the added benefit that we can determine the 
minimum BFS width for a graph in polynomial time.
Unfortunately, as the following theorem shows, bounded treewidth
does not imply bounded BFS width.

\begin{theorem}
There is a graph, $G$, with $n$ vertices
such that its treewidth $\tw(G)=1$ but $\bfsw(G)=n-1$ and $\bfsw_{\min}(G)=n-2$.
\end{theorem}
\begin{proof}
Let $G$ be the star graph with internal vertex $v$.
The BFS tree, $T$, from the vertex $v$ will have
$n-1$ vertices on its layer~$1$.
For any vertex $w\not=v$, 
the BFS tree, $T$, from $w$ will have $v$ 
on its layer~$1$, and the remaining $n-2$ vertices of $G$ on its layer 2.
\end{proof}

In these subsequent sections, we sometimes drop the graph, $G$, in 
width notations when the context is clear. When the graph is a tree with root $v$, a natural layering is defined by distances to the root and we use the term width to denote the BFS width of $G$ from $v$.
\section{Graphs of Bounded Bandwidth and Polylogarithmic BFS Width} 
\label{example}
In this section, we show how to construct a graph to 
prove that when bandwidth is bounded, 
BFS width can be as large as polylogarithmic in the number of vertices. 
This graph also proves that BFS width and minimum BFS 
width can differ by a logarithmic factor. 
The graph that we use to establish this lower bound
is an arbitrarily large binary tree.
Our construction is an induction
based on defining trees in terms of a level parameter, $k$,
which is the primary inductive term, 
plus an independent height parameter, $h$,
which determines the size, $n$, of the tree.

\textbf{Level 1.}
A level 1 tree, $T_{1,i}$, is simply a path, which defines
a tree when viewed as being rooted at one of 
the ends of this path.  Thus,
$\bw(T_{1,i}) = 1$ and $\bfsw(T_{1,i}) = 2$.
Before next jumping to the general case, for $k>1$, let us first describe
the level 2 trees, which introduce a constructive pattern we repeat for
the general case.

\textbf{Level 2.}
A level 2 tree  $T_{2, j}$  is defined for height, $h=2^j$, for any integer, $j>1$.
We start with $j+1$ level 1 trees, where $T_{1,i}$ has height $h_i = 2^{i} - 1$, 
for $i = 0, 1, \dots, j$. 
We also define another level 1 tree of height $h=2^j$, which
we call the \emph{spine}.
We connect each tree, $T_{1,i}$, with an edge to a vertex $v_i$ on the spine
so that every leaf in the resulting tree has the same depth.
See Figure~\ref{fig:bw2-bfs-log}.

\begin{figure}[hbt]
    \centering
    \includegraphics[width=.9\linewidth]{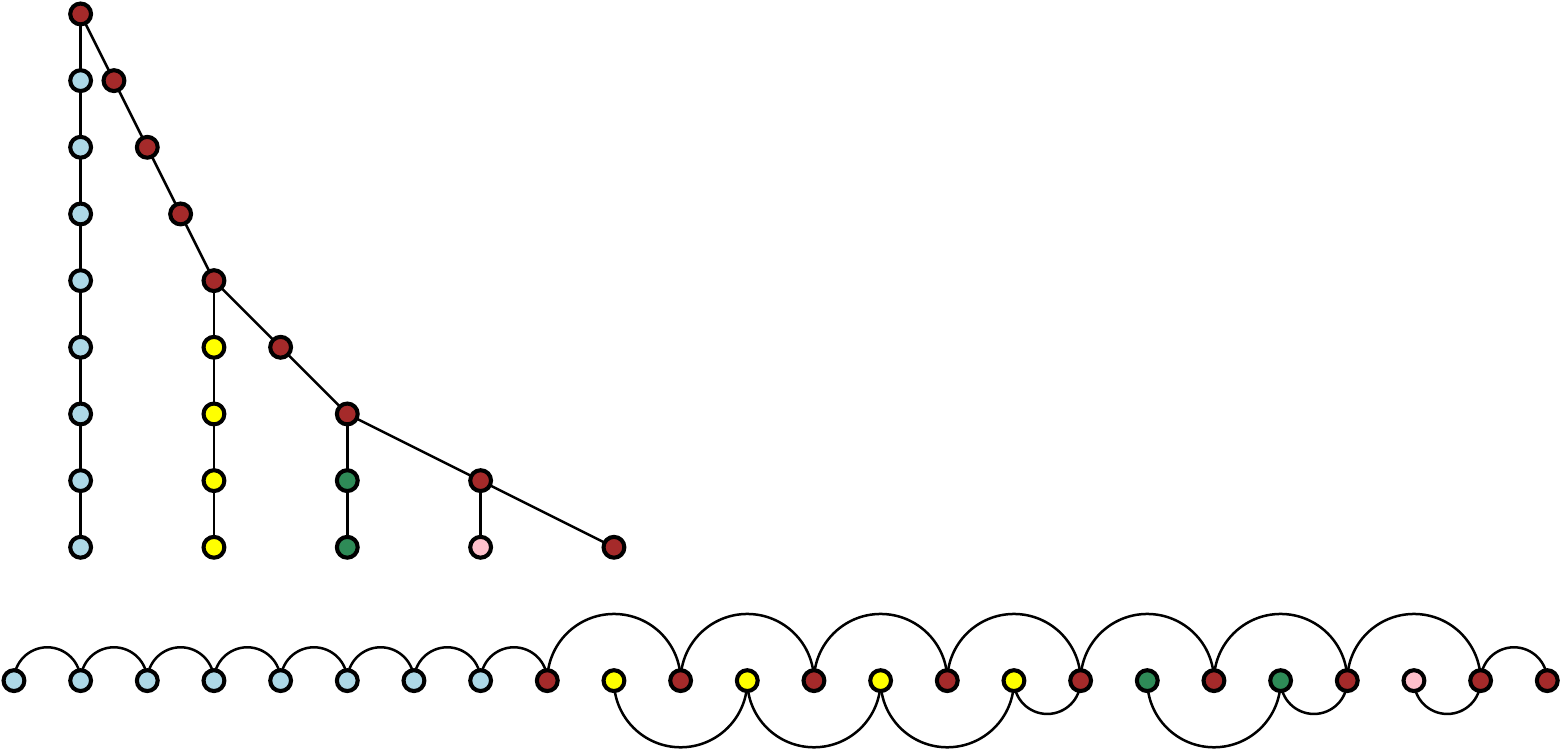}
    \caption{Top: the level 2 tree is constructed by attaching level 1 trees on a spine. The heights of these  level 1 subtrees approximately follow a geometric sequence.
    Bottom: an arc diagram illustrating the linear layout of the level 2 tree. Each edge is a semicircle. The maximum diameter of these semicircles equals the bandwidth of the linear layout. This linear layout achieves the minimum bandwidth 2. 
 }
    \label{fig:bw2-bfs-log}
\end{figure}

Thus, each subtree ``hanging off'' the spine 
is a path with height $h_i = 2^{i} - 1$, 
where $i = 0, 1, \dots, j$. The
height of the level 2 tree is $h  = 2^j$, where $j$
is independent of the level number 2.
As mentioned above,
all leaves in a level 2 tree have the same depth; hence, the number
of edges on the spine between two consecutive roots must be the
height difference between the corresponding subtrees:
$d(v_{i+1}, v_i) = h_{i+1} - h_i = 2^{i} = h_i + 1$, where $d$ is the shortest path distance.
Therefore, when we interleave a subtree, $T_{1,i}$,
and the spine edges between $v_{i+1}$ and $v_i$ in the linear layout, 
we can alternately take one vertex from each. 
This is illustrated in \cref{fig:bw2-bfs-log}. 
This layout has bandwidth 2, which is the minimum possible.

Let us now
consider the width (BFS width from the root) of a level 2 tree, $T_{2, j}$. This width lower bounds BFS width.
Since  $h  = 2^j$, $ j = \log h$ where $\log$ is base 2. The total number
of subtrees is $j + 1$, each contributing one vertex to the widest layer. The spine also has one vertex at the widest layer. 
Therefore, the width of  $T_{2, j}$ 
is $w = j + 1 + 1 = \log h + 2$.
We want to express this width in terms of the number of vertices, $n$. 
Any vertex of this tree is either in a subtree or on the spine.
The spine has height $h$ and $h + 1$ vertices. Subtrees have 
$n_s = \sum_{i=0}^{j} (h_i + 1)  = \sum_{i=0}^{j} 2^{i} = 2h - 1$ vertices. 
Thus, $n = h + 1 + 2h - 1 = 3h$;
hence, we can define a level 2 tree $T_{2, j}$ to have an arbitrarily large
number, $n$, of vertices.
Therefore,
$T_{2, j}$ has bandwidth $\bw(T_{2, j}) = 2$,
and width, $w = j + 2$, which implies
    $\bfsw(T_{2, j}) = \Omega(\log n)$,
where $n = 3h = 3 \cdot 2 ^ j$ can be arbitrarily large.

Following the structural pattern we used to build a level 2 tree,
to construct a tree at level $k$, we take trees from 
level $(k - 1)$ whose heights approximately follow a geometric sequence and
attach them to a path, which we again call the \emph{spine}. 
These level
$(k - 1)$ trees are called \emph{subtrees}. Their placement on the
spine is such that the leaves of these trees have the same depth
in the level $k$ tree, thus achieving their width in a single
(last) level of a breadth-first layering.
To get a bounded bandwidth
layout of the level $k$ tree, we interleave vertices from the spine and level
$(k - 1)$ subtrees. 
These constructions are 
illustrated in \cref{fig:bw-bfs-log}. 

\begin{figure}[hbt]
    \centering
    \includegraphics[width=.9\linewidth]{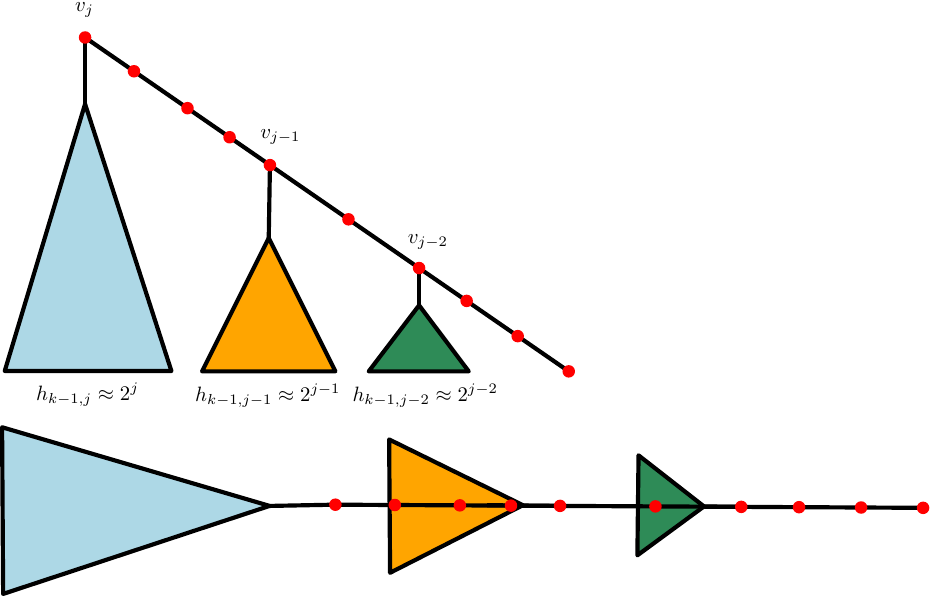}
    \caption{Top: the level $k$ tree is constructed by attaching level $(k - 1)$ trees on a spine. The heights of these  level $(k - 1)$ subtrees approximately follow a geometric sequence. The leaves of these subtrees have the same depth in the level $k$ tree.
    Bottom: a diagram illustrating the linear layout of the level $k$ tree. This does not necessarily achieve the minimum possible bandwidth. Intuitively, this is the level $k$ tree rooted at the rightmost vertex of the spine. We interleave vertices from the spine and level $(k - 1)$ subtrees. Vertices from a level $(k - 1)$ subtree are arranged according to their linear layout from the previous level.
    }
    \label{fig:bw-bfs-log}
\end{figure}





Let us, therefore, consider 
the bandwidth and BFS width of a level $k$ tree.
Suppose the
subtree below vertex $v_i$ (excluding $v_i$) is $T_{k-1,i}$.
Given this notational background,
let us prove the following theorem, which is one of the two main
technical theorems in this paper:

\begin{theorem}[Lower bound  on BFS width]
\label{theorem:lower}
    For each integer $k \ge 1$, and for arbitrarily 
large values of $n$, there exists a (level $k$) 
tree, $T$, with $n$ vertices that satisfies:
    \begin{align*}
            \bw(T) &= O_k(1), \\
    \bfsw(T) &= \Omega(\log^{k - 1} n) ,
    \end{align*}
    where $O_k(1)$ denotes a number that depends on $k$ 
but is bounded when $k$ is bounded.
\end{theorem}

\begin{proof}
We construct the level $k$ tree to have 
an arbitrarily large integer, $n$, of vertices, 
by induction, from level $(k-1)$ trees. 
To keep track of
the tree and subtree sizes, we let $n_{k,j}$ denote
the size of the tree, $T_{k,j}$, we are 
constructing (to roughly have height
that is arbitrarily large as a power of $j$).
The base case is the construction given above,
which in this notation, gives a level $2$ tree with $n_{2,j} = 3h_{2,j}$ vertices
and height $h_{2,j}=2^j$ and to
have bandwidth $2$ and BFS width at least $d_2j + 2$, for
the constant, $d_2=1$. 
Accordingly,
in our analysis below, we focus on the dominant components sufficient
for our asymptotic analysis. 
    
\textbf{Inductive step:}
Suppose, for $i=0,1,2,\ldots,j$,  
we have a level $(k-1)$ tree, $T_{k-1,i}$,
with height $h_{k-1,i} = 2^i + k - 3$ 
having width $w_{k-1,i}$ and $n_{k-1,i}$ vertices such that 
\begin{align*}
n_{k-1,i} = c_{k-1} \cdot 2^i + p_{k - 1} (i),   ~~~~
w_{k-1,i} =  d_{k-1} i^{k-2} + q_{k - 1} (i),  ~~~~
\bw'(T_{k - 1,i}) = O_{k - 1}(1),
\end{align*}
where $\bw'$ is the bandwidth of the linear layout we construct, $c_{k-1}$ and $d_{k-1}$ are constants, 
and $p_{k - 1}$ and $q_{k - 1}$ are polynomial functions with degree $ k - 3$.
Let us now construct
a level $k$ tree where the height of the leftmost subtree is 
$h_{k-1,j}=2^j + k - 3$, so
this level $k$ tree has height $h_{k,j} =  2^j + k - 2$.  
Using Faulhaber's formula 
(e.g., see Knuth~\cite{knuth1993johann}), we can characterize
the total number of vertices in all the subtrees as
\[
n_s \,=\, \sum_{i=0}^j( c_{k-1}  \cdot 2^i + p_{k - 1} (i) )\,=\, c_{k-1} \cdot (2^{j+1} - 1) + p_{k} (j),
\]
where $p_k$ is a polynomial function with degree $ k - 2$. 
The spine has $h_{k,j} + 1$ vertices.
The total number of vertices in this level $k$ tree, therefore, is
\begin{align*}
n_{k,j} &= n_s + h_{k,j} + 1 
    = c_{k-1} \cdot (2^{j+1} - 1)  + p_{k} (j) + 2^j + k - 2 + 1 \\
    &= (2c_{k-1} + 1) \cdot 2^j  + p_{k}' (j) 
    = c_k \cdot 2^j  + p_{k}' (j),
\end{align*}
where $p_k'$ is a polynomial function with degree $ k - 2$, and $c_k \coloneq 2c_{k-1} + 1$.
 Since $c_2 = 3$, we can solve this recurrence relation
to yield $c_k = 2 ^ k - 1$.

\textbf{Bandwidth:}
Let us next upper bound the bandwidth, $\bw(T_{k,j})$,
of the tree, $T_{k,j}$. 
We interleave each subtree $ T_{k-1,i}$ and the spine 
path between $v_{i+1}$ and $v_i$ in the linear layout.
$ T_{k-1,i}$ has height $h_{k-1,i} $ and $n_{k-1,i}$ vertices. The spine path $v_{i+1} v_i$ has length $d(v_{i+1}, v_i) = h_{k-1,i+1} - h_{k-1,i} = 2^{i}$.
We use the linear layout of subtree $ T_{k-1,i}$ obtained from the previous level and insert vertices from the spine path. 
Let $ r_i \coloneq n_{k-1,i} / 2^i = c_{k-1}  + p_{k-1}(i) / 2^i $.
We alternately place $\lceil r_i \rceil$ vertices from  $ T_{k-1,i}$  and one vertex from  the spine 
path  $v_{i+1} v_i$, until all vertices from  $ T_{k-1,i}$ have been placed, in which case we place the remaining vertices from $v_{i+1} v_i$ consecutively. Any edge connecting two vertices from $ T_{k-1,i}$ with length equal to 
$\bw'(T_{k - 1,i})$ may have at most $\left\lceil \frac{\bw'(T_{k-1,i})}{\lceil r_i \rceil + 1} \right\rceil$  spine vertices inside it. Any edge connecting two vertices from the spine path $v_{i+1} v_i$ has at most $\lceil r_i \rceil$ vertices from  $ T_{k-1,i}$ inside it.  These are the only two factors that cause the bandwidth to increase from level $(k-1)$ to level $k$, so the recurrence relation is 
\begin{align*}
    \bw'(T_{k,j}) \le \max_{i = 0}^j \max \left( \bw'(T_{k-1,i}) +  \left\lceil \frac{\bw'(T_{k-1,i})}{\lceil r_i \rceil + 1} \right\rceil, \lceil r_i \rceil + 1 \right)
    .
\end{align*}

Since $\lim_{i \to \infty} r_i = c_{k - 1}$, for large values of $i$, $r_i < c_{k - 1} + 1$, and for small values of $i$, $r_i$ can be upper bounded by another constant. Thus, this level $k$ linear layout has bandwidth $O_k(1)$.

\textbf{BFS width:}
BFS width is lower bounded by width, which can be computed using Faulhaber's formula 
(e.g., see Knuth~\cite{knuth1993johann}):
\begin{align*}
w_{k,j} = \sum_{i=0}^j (d_{k-1} \cdot i^{k-2} + q_{k - 1} (i))
    = \frac{d_{k-1}}{k - 1} j^{k-1} + q_{k} (j)
    = d_k j^{k-1}  + q_{k} (j), 
\end{align*}
where $d_k \coloneq \frac{d_{k-1}}{k - 1}$ and $q_k$ is a polynomial function with degree $k - 2$.
Since $d_2 = 1$, we can solve this recurrence relation
to yield $d_k = \frac{1}{ (k - 1)!}$.

We have therefore 
shown that the tree $T_{k,j}$ has
height $h_{k,j} =  2^j + k - 3$ and $n_{k,j}$ vertices, with
\begin{align*}
    n_{k,j} &= c_k \cdot 2^j  + p_{k}' (j),  \\
        \bw(T_{k,j}) &=
O_k(1), \mbox{~~~and}  \\
    \bfsw(T_{k,j}) &\ge w_{k,j}  = d_k j^{k-1} + q_{k} (j) = \Omega(\log^{k - 1} n_{k,j}).
\qedhere
\end{align*}
\end{proof}

Thus, a graph with bounded bandwidth can have 
polylogarithmic BFS width. Also:

\begin{theorem}[Lower bound on minimum BFS width]
\label{theorem:lower-bfsw-min}
For each integer $k \ge 1$, and for arbitrarily 
large values of $n$, there exists a (level $k$) 
tree, $T$, with $n$ vertices that satisfies:
    \begin{align*}
            \bw(T) &= O_k(1), \\
    \bfsw_{\min}(T) &= \Omega(\log^{k - 1} n) ,
    \end{align*}
    where $O_k(1)$ denotes a number that depends on $k$ 
but is bounded when $k$ is bounded.
\end{theorem}
\begin{proof}
A construction similar to \cref{fig:bw-bfs-log} can be used to place the root at one end of the linear layout. Taking two reflected copies of the resulting tree, connected at the root, gives a lower bound for minimum BFS width (with different constants than \cref{theorem:lower}). \end{proof}

\begin{theorem}
\label{theorem:ratio}
For arbitrarily large values of $n$, 
there exists a binary tree, $T$, 
with $n$ vertices that satisfies $\bfsw(T) / \bfsw_{\min}(T) = \Omega(\log n)$. 
\end{theorem}

\begin{proof}
The level 2 tree, $T = T_{2,j}$, is originally rooted at the leftmost spine vertex $v_j$. Consider rooting $T$ at the rightmost spine vertex $v_s$, as in the bottom of \cref{fig:bw2-bfs-log}. The width of the new tree is 2. Thus,
\begin{align*}
    \bfsw(T, v_j) &= j + 2,      \\
    \bfsw(T, v_s) &= 2,      \\
\frac{\bfsw(T)}{\bfsw_{\min}(T)} &\geq \frac{\bfsw(T, v_j)}{\bfsw(T, v_s)} 
= \frac{j + 2} {2} 
=  \frac{\log \left( \frac{n}{3} \right) + 2} {2}  , \mbox{~~~and} \\
\frac{\bfsw(T)}{\bfsw_{\min}(T)} &= \Omega(\log n).
\qedhere
\end{align*}
\end{proof}

\section{Polylogarithmic BFS Width for Graphs of Bounded Bandwidth}
In this section, we show that if bandwidth is bounded, then BFS width is at most polylogarithmic in the number of vertices. 
The construction in \cref{example} shows that this upper bound is tight.
The proof follows a structure similar to our lower-bound construction. 
First, we make several simplifying assumptions with respect to the graph,
which we make without loss of generality (wlog). 
Next, we identify spines and subtrees as in our construction above. 
Then we use strong induction to prove our bound.

The inductive hypothesis is as follows: for arbitrarily large values
of $n$, any graph on $n$ vertices with bandwidth $\bw \leq k$,
for a fixed constant, $k\ge 1$,
has BFS width $\bfsw = O\bigl(\log^{k - 1}(n)\bigr)$. 
The base case, $k=1$,
is clearly true, since a connected graph with bandwidth $1$ is a path.
For the induction step,
suppose that if
any graph, $G$, on $n$ vertices has bandwidth $\bw(G) \leq k$,
for a fixed constant, $k\ge 1$, then 
$\bfsw(G) = O\bigl(\log^{k - 1}(n)\bigr)$. 
To show the induction hypothesis holds for $k+1$, 
we consider a widest BFS tree of $G$, and show that there are at most a logarithmic
number of subtrees in $T$ and each subtree has width
$O\bigl(\log^{k - 1}(n)\bigr)$. Consequently, when these subtrees are
combined, the overall width remains polylogarithmic, with a 
higher (constant) exponent.

Suppose we have a linear layout of an $n$-vertex graph, $G$, that achieves 
minimum bandwidth,
let $T$ be a
widest BFS tree of $G$,
and let $r$ be the root of $T$.
We want to show that $T$ has width $ O\bigl(\log^{k}(n)\bigr)$.
Thus, without changing the BFS width, we can delete non-tree edges from $G$
as well as all vertices and edges below the first widest layer of $T$. 
This may reduce the number of vertices, but, wlog, let us nevertheless 
refer to the reduced number of vertices as $n$ and the reduced tree as $T$,
since we are interested in an upper bound in terms of the number of vertices
in the original graph. 
Thus, let us focus our attention on the reduced tree, $T$,
instead of the original graph.

Furthermore, we can make extreme points of the linear layout leaves of $T$. Suppose this is not the case. Wlog, we consider two scenarios:
\begin{enumerate}
    \item The leftmost vertex $v$  of the linear layout is the root. If the root only has one child, then we can remove it and make its child the new root. After repeating this step, either the new leftmost vertex is not the root, or it is the root with at least two children. In the latter case, we can take the subtree rooted at one child and flip its linear layout around the root. Thus, extreme points of the linear layout will not be the root.
    \item The rightmost vertex $v$  of the linear layout is neither the root nor a leaf. All leaves of the subtree $T_v$ rooted at $v$ are to the left of $v$ in the linear layout. We can flip this subtree around $v$ to move non-root vertices of $T_v$ to the right side of $v$. After this, we check the new rightmost vertex $w \in T_v$. If it is not a leaf, then $T_w \subset T_v$ and $v \notin T_w$. Since the subtree $T_w$ contains fewer vertices than $T_v$, each iteration of this process reduces the size of the subtree rooted at the rightmost vertex. Therefore, after repeating this process finitely many times, the rightmost vertex will be a leaf.

\end{enumerate}

After these operations, extreme points of the linear layout are now leaves of $T$. Note that the width of the tree does not change and the bandwidth of this new ordering cannot be greater than before.
The \emph{left spine} is defined as the path from the root to the leftmost vertex, and the 
\emph{right spine} is defined similarly. All non-spine vertices belong to a \emph{subtree} ``hanging off'' the spine. In this section, it is more convenient to treat the root vertices of these subtrees as the intersection of subtrees and the spine. Now we upper bound the number of subtrees that contribute to the width of $T$.
\begin{figure}[ht]  
    \centering
    \includegraphics[width=0.8\linewidth]{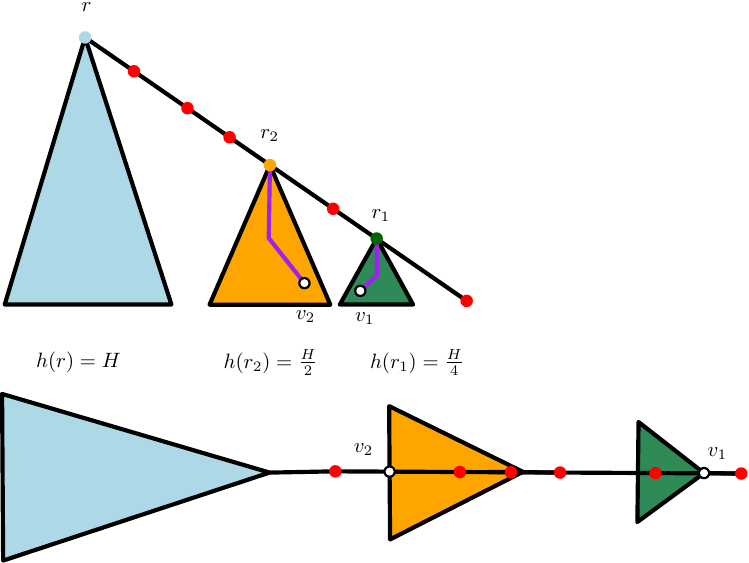}
    \caption{Proof of \cref{lemma:number} illustrated. When $i = 1$, $S_1$ is the set of all subtrees with height  $h(r_s) \in I \coloneq [\frac{H}{4}, \frac{H}{2}] $. In this figure, it contains the two  subtrees on the right. We want to upper bound  the total number of vertices in these subtrees.}
    \label{fig:lemma:number}
\end{figure}
\begin{lemma}\label{lemma:number}
    The number of subtrees that have at least one leaf at the widest layer of the tree $T$ is $O(\log n)$.
\end{lemma}
\begin{proof}
Because of symmetry, we only bound the number of subtrees whose root is on the right spine.   Suppose the tree $T$ rooted at $r$ has height $H$. For integer $i = 0, 1, \dots, \lfloor \log H \rfloor$, let $S_i$ be the set of all subtrees with height  $h(r_s) \in I \coloneq [\frac{H}{2 ^ {i + 1}}, \frac{H}{2 ^ {i}}] $ where $r_s$ is the root of a subtree. See \cref{fig:lemma:number}. We do not need to consider subtrees with height 0. 

Now we upper bound the total number of vertices in all subtrees in  $S_i$.
Consider the position of these subtrees in the linear layout. Suppose the rightmost vertex is $v_1$ and it belongs to the subtree rooted at $r_1$, and the leftmost vertex $v_2$ belongs to the subtree rooted at $r_2$. Recall that $d$ denotes the shortest path distance in the graph. For any root $r_s$ of a subtree in $ S_i$, $ d(r, r_s) = H - h(r_s) \in H - I \coloneq [H -\frac{H}{2 ^ {i}}, H - \frac{H}{2 ^ {i + 1}}]$. The shortest paths from $r$ to all such root vertices are part of the right spine. Thus,
\begin{align*}
d(r_1, r_2) &\leq  |H - I| = |I|  \\
    d(v_1, v_2) &\leq d(v_1, r_1) + d(r_1, r_2) + d(r_2, v_2) 
    \leq h(v_1) + |I| + h(v_2)  \\
  &\leq \frac{H}{2^i} + \left(\frac{H}{2^i} - \frac{H}{2^{i+1}}\right) + \frac{H}{2^i} 
    =  \frac{3H}{2^i} - \frac{H}{2^{i+1}} 
    = \frac{5H}{2^{i+1}}.
\end{align*}
Because of the bounded bandwidth, the number of vertices between $v_1$ and $v_2$ in the linear layout is at most $\frac{5H}{2^{i+1}} (\bw) $. 

Each subtree with height in $I$ must have at least $\frac{H}{2 ^ {i + 1}}$ vertices to reach the minimum height. Consequently, the total number of subtrees with height in any given interval $I$ is at most a constant and this completes the proof.
\end{proof}

Next, we prove that removing spine vertices from the linear layout reduces the bandwidth of subtrees. Since the roots of these subtrees also belong to the spine, each subtree may decompose into a forest. Given that bounded bandwidth implies bounded maximum degree, the number of trees in each resulting forest is bounded. If these trees have smaller bandwidth, we can apply the inductive hypothesis to show that they have width $O\bigl(\log^{k - 1}(n)\bigr)$. Thus, the original subtree also has width $O\bigl(\log^{k - 1}(n)\bigr)$. 
Thus,
the following lemma is essential to the induction:

\begin{lemma}
    If the tree $T$ has bandwidth $\bw$, then each tree in the forests obtained after removing the spine has bandwidth at most $\bw - 1$.
\end{lemma}
\begin{proof}
The proof outline is as follows. By definition, the extreme points of the original linear layout are spine vertices and thus subtrees with one vertex. We denote them as $v_l$ and $v_r$ respectively. Consider any subtree $T_v$ rooted at $v \neq v_l, v_r$. After removing the spine vertices, the linear layout for $T$ naturally defines a linear layout for  $T_v \setminus v$  by preserving the relative position of its vertices. We show that the maximum length of any edge of $T_v \setminus v$  will become less than $\bw$.  

In the linear layout for $T$, the length of any edge is at most $\bw$.  If an edge has length equal to $\bw$, we say that this edge is saturated. Let $S$ be the set of all saturated edges of $T_v \setminus v$. If this is an empty set, the proof is complete.
Otherwise, suppose there is one edge $e \in S$ with no spine vertex inside it.  $v_l$ and $v_r$ are connected through edges of the spine and one of these edges must enclose the two endpoints of $e$. The length of $e$ is $\bw$, so the length of this edge is greater than $\bw$, and this is a contradiction. Thus, there is a spine vertex inside every saturated edge of $T_v \setminus v$.   After removing the spine vertices, all edges of $T_v \setminus v$ have length at most $\bw - 1$. $T_v \setminus v$ may be a forest, so this upper bounds the bandwidth of each connected component.
\end{proof}

With these lemmas, we can apply strong induction on bandwidth, which gives
us:
\begin{theorem}[Upper bound on BFS width]
\label{theorem:upper}
For any constant, $k\ge 1$, and $n$-vertex graph, $G$, 
if $\bw(G) = k$, then $\bfsw(G) = O\bigl(\log^{k - 1}(n)\bigr)$.
\end{theorem}

\section{Applications}
In this section, we prove the application results claimed in the introduction.

\subsection{Numerical Linear Algebra}

For the purposes of this section, the only facts we need about the Cuthill–McKee and reverse Cuthill–McKee algorithms are:
\begin{itemize}
\item The bandwidth of a reordered matrix $PAP^T$, where $A$ is any symmetric matrix and $P$ is any permutation matrix, is the same as the bandwidth of the linear layout of the graph of nonzeros of $A$ obtained by laying out the vertices of this graph in the order given by permutation $P$.
\item Both the Cuthill–McKee and reverse Cuthill–McKee algorithms, applied to a matrix $A$ choose their orderings (interpreted as linear layouts of graphs) by the breadth-first layer of the graph of nonzeros of $A$, 
and then carefully permute the vertices within each layer.
\end{itemize}

We can directly relate the performance of these algorithms to BFS width:

\begin{theorem}
\label{thm:cm-bfsw}
Let $G$ be the graph of nonzeros of a symmetric matrix $A$, and let $v$ be an arbitrary vertex of $G$.
Then  the Cuthill–McKee and reverse Cuthill–McKee algorithms, using a breadth-first layering rooted at $v$,
produce layouts of bandwidth $\Theta\bigl(\bfsw(G,v)\bigr)$.
\end{theorem}

\begin{proof}
Let $u$ be a vertex in the widest layer of the breadth-first layering of $G$, chosen among the vertices of this layer to be the one that the Cuthill–McKee or reverse Cuthill–McKee algorithm places farthest from the previous layer. Then the positions of $u$ and its BFS ancestor differ by at least  $\bfsw(G,v)$,
so the bandwidth of the layout is at least $\bfsw(G,v)$. Every edge in $G$ extends over at most two consecutive layers, and has endpoints whose positions differ by at most $2\bfsw(G,v)-1$, so the bandwidth of the layout is at most $2\bfsw(G,v)-1$.
\end{proof}

By combining \cref{thm:cm-bfsw} with \cref{theorem:upper}, we have:

\begin{theorem}
\label{thm:cut1}
If a symmetric matrix $A$ can be reordered to have bounded bandwidth, the Cuthill–McKee and reverse Cuthill–McKee algorithms will produce a reordered matrix $PAP^T$ of at most polylogarithmic bandwidth.
\end{theorem}

Note that this theorem holds even if we do not know the optimal bandwidth
of the matrix $A$.
Moreover, the bound
of Theorem~\ref{thm:cut1} is
is tight, up to the dependence of the polylogarithmic bound on the bandwidth;
by combining \cref{thm:cm-bfsw} with \cref{theorem:lower-bfsw-min}, we have:

\begin{theorem}
\label{thm:cut2}
For each integer $k \ge 1$, there exist $n\times n$ symmetric matrices $A$ of bounded bandwidth (for arbitrarily large $n$) for which the Cuthill–McKee and reverse Cuthill–McKee algorithms produce reordered matrices $PAP^T$ of bandwidth $\Omega(\log^{k - 1} n)$.
\end{theorem}

Theorems~\ref{thm:cut1} and~\ref{thm:cut2} provide the first
non-trivial deterministic analyses for the Cuthill-McKee algorithm
and its reversal.

\subsection{Graph Reconstruction}

To reconstruct graphs of low bandwidth or low BFS width, we apply 
\cref{algorithm1}.
 
\begin{algorithm}[hbt]
\caption{Reconstructing graphs with low BFS width}
\label{algorithm1}
\begin{algorithmic}[1]
\State $v \leftarrow$ an arbitrary vertex in $V$    \Comment{this is the root of the BFS tree}
\For{$u \in V\setminus\{v\}$}
    \State Query$(v,u)$  \Comment{this is a distance query}
\EndFor
\State $\hat{E} \leftarrow$ $\bigl\{ \{a, b\} : \{a, b\} \subseteq V\setminus\{v\} \land |d(v, a) - d(v, b)| \leq 1\bigr\}$
\For{$\{a, b\} \in \hat{E}$}
    \State Query$(a, b)$
\EndFor
\State \Return $\bigl\{\{u,v\}: u\in V\land d(v,u)=1\bigr\}\cup \bigl\{\{a, b\}: \{a, b\}\in\hat{E}\land d(a,b)=1\bigr\}$
\end{algorithmic}
\end{algorithm}

Less formally, \cref{algorithm1} first picks a vertex $v$ as the root of the BFS tree. Then, it obtains distances from this root to every other vertex to determine their layer number. It constructs the set of all vertex pairs from the same layer or in two consecutive layers, and queries each such pair. Recall that in the proof of \cref{thm:forward}, non-tree edges connect two vertices whose layer numbers differ by at most one, so each pair of adjacent vertices in $G$ must belong to the set of queried pairs. The algorithm then returns the set of all pairs for which a query found two vertices to be at distance one.
\cite{bastide_tight_2025} used a similar approach to reconstruct trees and $k$-chordal graphs. They divided the graph into BFS layers and reconstructed edges  in the same layer and in two consecutive layers, but querying each pair would incur a high query complexity for those graph classes.

\begin{theorem}
\label{thm:reconstruct}
For any graph $G$ with $n$ vertices and BFS width $B$,  \cref{algorithm1} reconstructs $G$ with deterministic query complexity $O(nB)$.
\end{theorem}

\begin{proof}
It takes $n-1$ queries to find the distances from the chosen root $v$ and partition the vertices into layers.
Then, the algorithm makes at most $3B-1$ queries involving each remaining vertex, between it and at most $3B-1$ other vertices in its layer and the two adjacent layers. This double-counts the queries, so the total number of queries is at most $n-1+(n-1)(3B-1)/2=O(nB)$.
\end{proof}

\begin{corollary}
For any constant bound $b$ on the bandwidth of graphs, \cref{algorithm1} reconstructs $n$-vertex graphs of bandwidth $\le b$ with deterministic query complexity $\tilde{O}(n)$.
\end{corollary}

\begin{proof}
This follows immediately from \cref{thm:reconstruct} and from the relation between bandwidth and BFS width of \cref{theorem:upper}.
\end{proof}

\subsection{Graph Drawing}
In our graph drawing application, we specifically
focus on arc diagrams, which were used to study the crossing number
problem \cite{masuda_crossing_1990} and visualize structure in
strings \cite{wattenberg_arc_2002}. 
In an arc diagram, vertices
of a graph are positioned along a straight (e.g., horizontal) line, 
and edges are
semicircles drawn on either side of the line \cite{masuda_crossing_1990},
joining pairs of points (possibly using straight-line segments for consecutive
points).
Examples are shown above in Figures~\ref{fig:arcs} and~\ref{fig:bw2-bfs-log}.  

To analyze such diagrams, we constrain the vertices to be placed at
consecutive integer coordinates. Thus, the horizontal line is a
linear layout of the vertices. The drawing fits within a bounding box $n-1$ units wide and $b$ units tall, where $n$ is the number of vertices in the graph and $b$ is the bandwidth of the given layout.

We immediately obtain the following result:

\begin{theorem}
Suppose that a given graph $G$ has $n$ vertices and has an arc diagram of bounded height $h$. Construct an arc diagram for $G$ by ordering the vertices by their breadth-first layer numbers, starting from an arbitrary root vertex. Then if $h$ is bounded by a constant, the height of the arc diagram that we construct will be bounded by $O\bigl(\polylog(n)\bigr)$.
\end{theorem}

\bibliography{references, bib2doi}
\end{document}